\documentclass[conference]{IEEEtran}
\IEEEoverridecommandlockouts
\usepackage{amsthm}
\usepackage{calc}

\newtheorem{lemma}{Lemma}
\newtheorem{claim}{Claim}

\newtheorem{example}{Example}
\newtheorem{remark}{Remark}
\newtheorem{definition}{Definition}
\usepackage{graphicx,cite}

\usepackage{blindtext, graphicx, amsfonts,
	amssymb,multirow}

\ifCLASSINFOpdf

\else

\fi

\usepackage[cmex10]{amsmath}

\usepackage{algorithm}
\usepackage{algorithmic}

\newcommand{\calA}{\mathcal{A}}

\newcommand{\calN}{\mathcal{N}}
\newcommand{\calP}{\mathcal{P}}

\newcommand{\calI}{\mathcal{I}}

\newcommand{\bfu}{\mathbf{u}}
\newcommand{\bfc}{\mathbf{c}}
\newcommand{\bfT}{\mathbf{T}}
\newcommand{\bfG}{\mathbf{G}}

\begin{document}
	%
	\title{Coded Caching with Low Subpacketization Levels}

	\author{\IEEEauthorblockN{Li Tang and Aditya Ramamoorthy}
		\IEEEauthorblockA{Department of Electrical and Computer Engineering\\
			Iowa State University\\
			Ames, IA 50010\\
			Emails:\{litang, adityar\}@iastate.edu}}

	
	%


	\maketitle

\begin{abstract}

Caching is popular technique in content delivery networks that allows for reductions in transmission rates from the content-hosting server to the end users. Coded caching is a generalization of conventional caching that considers the possibility of coding in the caches and transmitting coded signals from the server. Prior results in this area demonstrate that huge reductions in transmission rates are possible and this makes coded caching an attractive option for the next generation of content-delivery networks. However, these results require that each file hosted in the server be partitioned into a large number (i.e., the subpacketization level) of non-overlapping subfiles. From a practical perspective, this is problematic as it means that prior schemes are only applicable when the size of the files is extremely large. In this work, we propose a novel coded caching scheme that enjoys a significantly lower subpacketization level than prior schemes, while only suffering a marginal increase in the transmission rate. In particular, for a fixed cache size, the scaling with the number of users is such that the increase in transmission rate is negligible, but the decrease in subpacketization level is exponential.


\end{abstract}
	

	%

	\section{Introduction}
Fast and efficient content delivery over the Internet is an important problem and is the core business of companies such as Akamai (which is estimated to serve 15-30\% of all Web traffic). Crucial to Akamai's approach is a large network of servers that cache popular content closer to the end users. This serves to significantly reduce delivery time and improve the end user's experience. Traditional caching operates by storing popular content (or portions thereof) closer to or at the end user. Typically a cache serves a user request partially (or sometimes entirely) with the remainder of the content coming from the server.
	

Prior work in this area \cite{Ma14} demonstrates that allowing coding in the cache and coded transmission from the server (referred to as {\it coded caching}) to the end users can allow for huge reductions in the number of bits transmitted from the server to the end users. This is an exciting development given the central role of caching in supporting a significant fraction of Web traffic.


Reference \cite{Ma14} considered a scenario where a single server containing $N$ files of size $F$ bits connects to $K$ users over a shared link and each user has a cache memory $MF$ bits.
Coded caching consists of two distinct phases: a $\emph{placement phase}$ and a $\emph{delivery phase}$. In the placement phase, the user caches are populated. This phase does not depend on the user demands which are assumed to be arbitrary. In delivery phase, server sends a $\emph{coded}$ signal to each user such that each user's demand is satisfied.

There have been subsequent papers in this area. Several papers \cite{ghasemi2015improved, HG15, sengupta2015improved}, have considered the problem of tighter lower bounds on the coded caching rate. Several variants of the problem have been examined. The case when files have different popularity levels has been examined in \cite{maddahN14nonuniform_demand,jiTLC14zipf,hachemKD14a}, device-to-device (D2D) wireless networks where there is no central server were considered in \cite{jiCM13,senguptabeyondd2d} and systems with differing file sizes were investigated in \cite{zhang2015coded}. Coded caching over a more general class of network topologies was examined in \cite{JJ15,tangR16}.

In this work we investigate certain issues with the achievability scheme of \cite{Ma14}. In particular, in the placement phase of \cite{Ma14} each file is split into a large number of subfiles (henceforth, the subpacketization level); the number of subfiles grows exponentially with $K$ for a fixed cache size. This can cause issues in actual implementations of coded caching.
Specifically, even for moderate number of users ($K$), the size of the files stored in the server need to be very large. Moreover, in practice each subfile needs to have appropriate header information that allows for bookkeeping at the server and the users. The rate overhead associated with the header will also grow as the number of subfiles is large. We discuss this issue in more detail in Section \ref{sec:prob_form}.

In this work, we propose new schemes for coded caching that have significantly smaller subpacketization level than the scheme of \cite{Ma14}. Our schemes are derived from constructions of combinatorial objects known as resolvable designs \cite{Stinson} in the literature. This issue was considered in the work of \cite{shanmugam_et_al14}, but for the case of decentralized caching. In independent work, \cite{yan_et_al16} arrived at a similar result to the one presented in our paper. However, the techniques used in our paper are quite different and our construction is significantly simpler than theirs.

\subsection{Main contributions}
\begin{itemize}
\item The subpacketization level of our scheme is exponentially lower than the scheme of \cite{Ma14}. This implies that our schemes are much more amenable to practical implementations even for smaller values of $K$
\item The transmission rate of our scheme is not too much higher than the scheme of \cite{Ma14}. In particular, for large $K$, both schemes have almost the same rate.
\end{itemize}

This paper is organized as follows, Section \ref{sec:prob_form} presents the problem formulation and preliminary definitions. In Section \ref{sec:low_subpack}, we
describe our proposed coded caching scheme and analyze its performance. We compare the performance of our scheme with competing schemes in Section \ref{sec:comp}.  Section \ref{sec:conclusion} concludes the paper with a discussion of future work.

\section{Problem Formulation and Preliminaries}
\label{sec:prob_form}
	
	In this work, we consider a caching system consisting of a single server and $K$ users, $U_1,\cdots, U_K$, such that the server is connected to all the users through an error-free shared link (see Fig. \ref{Fig:Cachingsystem}). The server contains a library of $N$ files where each file is of size $F$ bits. These files are represented by random variables $W_i$, $i=1,\cdots, N$, where $W_i$ is distributed uniformly over the set $[2^F]$ (we use $[n]$ to denote the set $\{1, 2, \dots, n\}$ throughout). Each user has a cache memory of $MF$ bits, where $M \leq N$.
	
	\begin{figure}[t]
		\centering
		\includegraphics[scale=0.64]{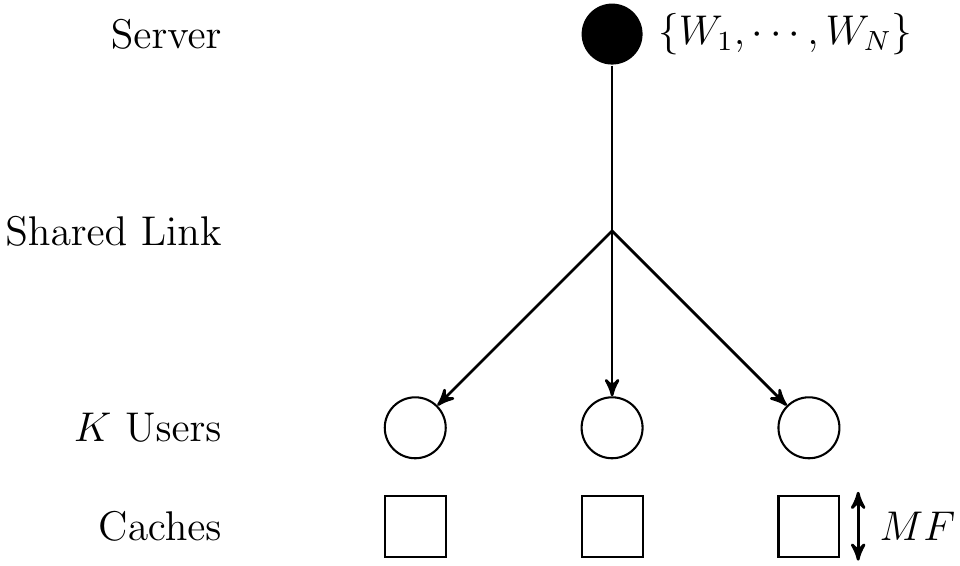}
		\caption{A coded caching system with $N$ files in the server, $K$ users each equipped with a cache of size $MF$ bits. The server communicates with the users over a shared link.}
		\label{Fig:Cachingsystem}
		\vspace{-0.2in}
	\end{figure}
	
An $(M,R)$ caching system can be defined as follows.	
    \begin{itemize}
    	\item
    	\emph{$K$ caching functions:} User $U_i$ caches $Z_i=\phi_i(W_1,\cdots,W_N)$ in the placement phase. Here, $\phi_i: [2^{F}]\to [2^{MF}]$.
    	\item
    	\emph{$N^K$ encoding functions:} The server transmits signals $X_{d_1,\cdots,d_K}=\psi_{d_1,\cdots,d_K}(W_1,\cdots,W_N)$ over the shared link to each user in the delivery phase. Here, $\psi_{d_1,\cdots,d_K}:[2^{NF}]\to[2^{RF}]$.
    	\item
    	\emph{$KN^K$ decoding functions:} User $U_i$ uses the decoding function $\hat{W}_{d_1,\cdots,d_K;i}=\mu_{d_1,\cdots,d_K;i}(X_{d_1,\cdots,d_K}, Z_i)$. Here, $\mu_{d_1,\cdots,d_K;i}(X_{d_1,\cdots,d_K}, Z_i):[2^{RF}]\times[2^{MF}]\to[2^{F}]$
    \end{itemize}
	
	The probability of error in a coded caching system is defined as $P_e=\max_{(d_1, \dots, d_K) \in [N]^K}\max_{i \in [K]} P(\hat{W}_{d_1,\cdots,d_K;i}\neq W_{d_i})$. The pair $(M,R)$ is said to be achievable if for every $\epsilon > 0$ and every large enough file size $F$, there exists a $(M,R)$ caching system such that $P_e$ is at most $\epsilon$.
	
The work of Maddah-Ali and Niesen \cite{Ma14} proposed an achievable $(M,R)$ caching scheme when $t = KM/N$ is an integer. Their scheme partitions each file into $F_s=\binom{K}{\frac{KM}{N}}$ non-overlapping subfiles of equal size and places judiciously chosen subsets of the subfiles into the user caches in the placement phase. This is referred to as uncoded placement in the literature. They achieve a rate of
	$$R=K\times\bigg(1-\frac{M}{N}\bigg)\times \frac{1}{1+\frac{KM}{N}},$$ when $K\le N$. 
It can be observed that for the scheme of \cite{Ma14}, the subpacketization level $F_s$ grows exponentially with $K$, when $M/N$ is fixed. This can be problematic in practical implementations.

For instance, the atomic unit of storage on present day hard drives is a sector of size $512$ bytes and the trend in the disk drive industry is to move this to 4096 bytes. Now, suppose that $K=50$, with $\frac{M}{N}=\frac{1}{2}$ so that $F_s\approx 10^{14}$. In this case, it is evident that one needs the files to be of at least size $\approx 5 \times 10^8$ gigabytes for leveraging the gains promised by the scheme of \cite{Ma14}. Thus, their scheme is not practical in this setting. Even for smaller values of $K$, schemes with low subpacketization levels are desirable. This is because any practical scheme will require each of the subfiles to have some header information that allows for decoding at the end users. When there are a large number of subfiles, the header overhead may be non-negligible.


In this work, we propose novel placement and delivery schemes that operate with significantly lower subpacketization levels. Towards this end, we first demonstrate that uncoded placement where each user has the same amount of cache memory can be represented by a block design \cite{Stinson}.

\begin{definition}
A design is a pair $(X, \calA)$ such that
\begin{enumerate}
\item  $X$ is a set of elements called points, and
\item $\calA$ is a collection (i.e., multiset) of nonempty subsets of $X$ called blocks, where each block contains the same number of points.
\end{enumerate}
\end{definition}

A design is in one-to-one correspondence with an incidence matrix $\calN$ which is defined as follows.
\begin{definition}
The incidence matrix $\calN$ of a design $(X, \calA)$ is a binary matrix of dimension $|X| \times |\calA|$, where the rows and columns correspond to the points and blocks respectively.
Let $i \in X$ and $j \in \calA$. Then,
\begin{align*}
\calN (i,j) = \begin{cases}
1 & \text{~if $i \in j$,}\\
0 & \text{~otherwise}.
\end{cases}
\end{align*}
\end{definition}
In general, we can define the placement schemes by using the incidence matrix. One can view the placement scheme of \cite{Ma14} when $KM/N = t$ is an integer as an instance of a block design as follows. We associate the users with the points, i.e., $X = [K]$ and the subfiles as the blocks, i.e., $\calA = \{B: B \subset [K], |B| = t\}$. Each file $W_n$ is divided into $\binom{K}{t}$ parts indexed as $W_{n,B}, B \in \calA$. User $i$ caches $W_{n,B}$ for $B \in \calA$ if $i \in B$ or equivalently if the corresponding entry in the incidence matrix is a one. In general, we can reverse the roles of the points and blocks and choose to associate the users with the blocks and subfiles with the points instead. The transpose of the incidence matrix then allows us to specify the placement.

In this work, we will utilize resolvable designs which are a special class of block designs.
\begin{definition}
\label{defn:resolv_design}
A parallel class $\calP$ in a design $(X,\calA)$ is a subset of disjoint blocks from $\calA$ whose union is $X$. A partition of $\calA$ into several parallel classes is called a resolution, and $(X,\calA)$ is said to be a resolvable design if $\calA$ has at least one resolution.
\end{definition}

We now provide an example of how a design can be used in the placement scheme, when the users and subfiles correspond to the blocks and points, respectively.	
\begin{example}
Consider a block design specified as follows.
\begin{align*}
X&=\{1,2,3\}, & \calA &=\{\{1,2\},\{1,3\},\{2,3\}\}, \text{~with}\\
 \calN &= \begin{bmatrix}
1 & 1 & 0\\
1 & 0 & 1\\
0 & 1 & 1
\end{bmatrix}.
\end{align*}
As described below, it corresponds to a coded caching scheme with $K=3$ and $M/N=2/3$.

We let the blocks correspond to users which are denoted as $U_{12}, U_{13}, U_{23}$. Each file is subdivided into $|X|=3$ subfiles denoted as $W_{n,1}, W_{n,2}, W_{n,3}$ for $n \in [N]$. The placement is specified as follows.
\begin{align*}
Z_{12}&=\{W_{n,1},W_{n,2}\}_{n=1}^{N}\\
Z_{13}&=\{W_{n,1},W_{n,3}\}_{n=1}^{N}\\
Z_{23}&=\{W_{n,2},W_{n,3}\}_{n=1}^{N}.
\end{align*}
In the delivery phase, suppose that $U_{12}, U_{13}, U_{23}$ request files $W_{d_{12}}, W_{d_{13}}, W_{d_{23}}$. Using the delivery signal
 \begin{align*}
 W_{d_{12},3}\oplus  W_{d_{13},2}\oplus W_{d_{23},1}
 \end{align*}
all three users can recover their missing subfiles.
	\end{example}
	
\section{A Low Subpacketization Level Scheme}
\label{sec:low_subpack}
Consider a coded caching scenario where the number of users $K$ can be factored as $K=q \times k$ (this requires $K$ to be composite). In this section we use resolvable designs to arrive at a scheme where the subpacketization level is significantly smaller than prior schemes. 

\subsection{Resolvable Design Construction}
\label{sec:resolv_design_constr}

Let $\mathbb Z_q$ denote the additive group of integers modulo $q$. Consider the generator matrix of a $(k,k-1)$ single parity check (SPC) code over $\mathbb Z_q$ defined below.
\begin{equation}
\bfG_{SPC}=
\begin{bmatrix}
& &\vline&1\\
&\huge \bf I_{k-1}&\vline&\vdots\\
&&\vline&1
\end{bmatrix}.
\label{equ:SPC}
\end{equation}
This code has $q^{k-1}$ codewords which can be obtained by computing $\bfc = \bfu \cdot \bfG_{SPC}$ for all possible message vectors $\bfu$.
We collect the $q^{k-1}$ codewords $\bfc_i$ and construct a matrix $\bfT$ of size $k\times q^{k-1}$ specified as follows.
\begin{equation}
\bfT=[{\bfc}_1^T,{\bfc}_2^T,\cdots,{\bfc}_{q^{k-1}}^T].
\label{equ:T}
\end{equation}
Let $X_{SPC} = [q^{k-1}]$ represent the point set of the design.
We define the blocks as follows. For $0 \leq l \leq q-1$, let $B_{i,l}$ be a block defined as
$$
B_{i,l}=\{j: \mathbf T_{i,j}=l\}.
$$
The set of blocks $\calA_{SPC}$ is given by the collection of all $B_{i,l}$ for $1 \leq i \leq k$ and $0 \leq l \leq q-1$ so that $|\calA_{SPC}| = kq$.

\begin{example}
\label{eg:resolv_q_2_k_3}
	Let $q=2$, $k=3$. Consider a $(3,2)$ SPC code over $\mathbb Z_2$ with generator matrix
	$$
   \mathbf G_{SPC}=\begin{bmatrix}
	1&0&1\\
	0&1&1
	\end{bmatrix}.
	$$
    The four codewords in this code are
	$\bfc_1=[0 ~0 ~0]$,
	$\bfc_2=[0 ~1 ~1]$,
	$\bfc_3=[1 ~0 ~1]$,
	$\bfc_4=[1 ~1 ~0]$, and $\bfT$ is constructed as follows.
	$$
	\bfT=
	\begin{bmatrix}
	0&0&1&1\\
	0&1&0&1\\
	0&1&1&0
	\end{bmatrix}.
	$$
Using $\bfT$, we generate the resolvable block design $(X,\calA)$ as follows.
The point set $X=\{1,2,3,4\}$. Block $B_{1,0}$ is obtained by determining the column indexes where the first row of $\bfT$ is zero. Thus $B_{1,0} = \{1,2\}$.
Proceeding in this manner we obtain
$$\calA=\{\{1,2\},\{3,4\},\{1,3\},\{2,4\},\{1,4\},\{2,3\}\}.$$
It can be observed that $\calA$ has a resolution ({\it cf.} Definition \ref{defn:resolv_design}) with the following parallel classes.
\begin{align*}
\calP_1 &=\{\{1,2\},\{3,4\}\},\\
\calP_2 &=\{\{1,3\},\{2,4\}\}, \text{~and}\\
\calP_3 &=\{\{1,4\},\{2,3\}\}.	
\end{align*}
\end{example}

The following lemma shows that the construction procedure above always results in a resolvable design.
\begin{lemma}
The construction procedure above produces a design $(X_{SPC}, \calA_{SPC})$ where $X_{SPC} = [q^{k-1}]$, $|B_{i,l}| = q^{k-2}$ for all $1 \leq i \leq k$ and $0 \leq l \leq q-1$. Furthermore, the design is resolvable with parallel classes given by $\calP_i = \{B_{i,l}: 0 \leq l \leq q-1\}$, for $1\leq i \leq k$.
\end{lemma}
\begin{proof}
For a given $i$, we need to show that $|B_{i,l}| = q^{k-2}$ for all $0 \leq l \leq q-1$ and that $\cup_{l=0}^{q-1} B_{i,l} = [q^{k-1}]$. Towards this end we note that for $\Delta = [\Delta_1 ~\Delta_2~ \dots~\Delta_k]= \bfu \bfG_{SPC}$, we have
\begin{align*}
\Delta_i = \begin{cases}
\bfu_i & i = 1, \dots, k-1,\\
\sum_{j=1}^{k-1} \bfu_j & i=k.
\end{cases}
\end{align*}
Thus, for $1 \leq i \leq k-1$, we have that $|B_{i,l}| = |\{\bfu : \bfu_i = l\}|$ which in turns equals $q^{k-2}$ as it is the subset of all message vectors with the $i$-th coordinate equal to $l$. Moreover, as the $i$-th coordinate has to belong to $\{0, \dots, q-1\}$, we have that $\calP_i = \{B_{i,l}: 0 \leq l \leq q-1\}$ forms a parallel class.

It remains to show the same result when $i = k$. For this consider the equation
\begin{align*}
\sum_{j=1}^{k-2} \bfu_j = l - \bfu_{k-1}
\end{align*}
where $l$ is fixed.
For arbitrary $\bfu_j, 1 \leq j \leq k-2$, this equation has a unique solution for $\bfu_{k-1}$.
This implies that for any $l$, $|B_{k,l}| = q^{k-2}$ and that $\calP_{k}$ forms a parallel class.
\end{proof}
\begin{remark}
If $q$ is a prime power then constructions of affine resolvable balanced incomplete block designs (BIBDs) with significantly more parallel classes are known (see \cite{Stinson}, Ch. 5). However, our proposed scheme above works for any value of $q$ and is adapted for the application to coded caching that we consider.
\end{remark}
\subsection{Usage in a coded caching scenario}
We first demonstrate our proposed placement scheme by using Example \ref{eg:placement_q_2_k_3} below. We associate the users with the blocks and subfiles with the points of the design.
\begin{example}
\label{eg:placement_q_2_k_3}
Consider the resolvable design from Example \ref{eg:resolv_q_2_k_3}. The six blocks in $\calA$ correspond to six users $U_{12}$, $U_{34}$, $U_{13}$, $U_{24}$, $U_{14}$, $U_{23}$. Each file is partitioned into $F_s=4$ subfiles $W_{n,1}, W_{n,2}, W_{n,3}, W_{n,4}$ which correspond to the four points in $X$. The cache in user $U_{B}$, denoted $Z_{B}$ is specified as
	\begin{align*}
	Z_{12}=&(W_{n,1}, W_{n,2})_{n=1}^{N}\\
	Z_{34}=&(W_{n,3}, W_{n,4})_{n=1}^{N}\\
	Z_{13}=&(W_{n,1}, W_{n,3})_{n=1}^{N}\\
	Z_{24}=&(W_{n,2}, W_{n,4})_{n=1}^{N}\\
	Z_{14}=&(W_{n,1}, W_{n,4})_{n=1}^{N}\\
	Z_{23}=&(W_{n,2}, W_{n,3})_{n=1}^{N}
	\end{align*}
This corresponds to a coded caching system where each user caches half of each file so that $M/N = 1/2$.
\end{example}
Suppose that in the delivery phase user $U_B$ requests file $W_{d_{B}}$ where $d_B \in [N]$. These demands can be satisfied as follows.
\begin{example}
	Consider the placement scheme specified in Example \ref{eg:placement_q_2_k_3}.  For a set of requests $W_{d_{12}},W_{d_{34}},W_{d_{13}},W_{d_{24}},W_{d_{14}}$ and $W_{d_{23}}$, we pick three blocks from three different parallel classes $\calP_1$, $\calP_2$, $\calP_3$ and generate the signals transmitted in the delivery phase as follows.
	\begin{align*}
	&W_{d_{12},3}\oplus W_{d_{13},2}\oplus W_{d_{23},1},\\
    &W_{d_{12},4}\oplus W_{d_{24},1}\oplus W_{d_{14},2},\\
    &W_{d_{34},1}\oplus W_{d_{13},4}\oplus W_{d_{14},3}, \text{~and}\\
    &W_{d_{34},2}\oplus W_{d_{24},3}\oplus W_{d_{23},4}.
	\end{align*}
The three sums in the first signal correspond to blocks from different parallel classes $\{1,2\}\in \calP_1, \{1,3\}\in \calP_2 ,\{2,3\}\in \calP_3$. It can be observed that this equation benefits each of the three users participating in it. Furthermore, it is also apparent that at the end of the delivery phase, each user obtains its missing subfiles. This scheme corresponds to a subpacketization level of $4$ and a rate of $1$. In contrast, the scheme of \cite{Ma14} would require a subpacketization level of $\binom{6}{3} = 20$ with a rate of $0.75$.
\end{example}

Upon inspection, it can be observed that the proposed scheme works since it allows us to always generate an equation where one user from each parallel class can participate. Crucially, the equations can be chosen so that at the end of the transmission each user is satisfied.


The basic idea conveyed by the example above can be generalized as follows. For a coded caching scheme with $K = kq$ and $M/N = 1/q$, suppose that we generate the resolvable design $(X, \calA)$ by the procedure outlined in Section \ref{sec:resolv_design_constr}. Let each block in $\calA$ correspond to a user and each point in $X$ correspond to a subfile. We split each file $W_n, n \in [N]$ into $q^{k-1}$ subfiles, so that $W_n=\{W_{n,t}:t\in [q^{k-1}]\}$ and perform the cache placement by using the incidence matrix of the design. Thus, subfile $W_{n,t}$ is placed in the cache of user $U_B$ if $t\in B$ and hence each user caches a total of $Nq^{k-2}$ subfiles. Since each of subfiles has size $\frac{F}{q^{k-1}}$, this requires
\begin{align*}
Nq^{k-2}\frac{F}{q^{k-1}}&=F\frac{N}{q}=FM
\end{align*}
bits of cache memory at each user. Thus, the memory constraint is satisfied.

It remains to show that we can design a delivery phase scheme that satisfies any possible demand pattern. Towards this end we need the following claim, whose proof is deferred to the Appendix.
\begin{claim}
\label{claim:intersect}
Consider a resolvable design $(X, \calA)$ constructed by the procedure in Section \ref{sec:resolv_design_constr} for given $k$ and $q$. Let the parallel classes of the design be denoted $\calP_1, \dots, \calP_k$. Consider blocks $B_{i_1, l_1}, \dots, B_{i_{k-1}, l_{k-1}}$ (where $i_j \in [k], l_j \in \{0, \dots, q-1\}$) that are picked from $k-1$ distinct parallel classes $\calP_{i_1}, \dots, \calP_{i_{k-1}}$. Then, $|\cap_{j=1}^{k-1} B_{i_j, l_j}| = 1$.
\end{claim}
Note that the blocks are in one to one correspondence with the users. Thus, Claim \ref{claim:intersect} shows that any $k-1$ users picked from different parallel classes have one subfile in common.

Roughly speaking, this implies that if we pick $k$ users, one from each parallel class, then we can generate an equation that is simultaneously useful to each of them. A subsequent counting argument shows that each user can be satisfied.

We now formalize this argument. Let the request of user $U_B, B \in \calA$ be denoted by $W_{d_{B}}$.\\
{\bf Delivery Phase algorithm}
\begin{enumerate}
\item Pick users $U_{B_{1, l_1}}, \dots, U_{B_{k, l_k}}$  where $l_i \in \{0, \dots, q-1\}, i \in [k]$ and $B_{i, l_i} \in \calP_i$, such that $\cap_{i=1}^k B_{i,l_i} = \phi$.
\item Let $\hat{l}_{\alpha} = \cap_{i \in [k] \setminus \{\alpha\} } B_{i, l_i}$ for $\alpha = 1, \dots, k$.
\item Server transmits
\begin{align*}
\oplus_{\alpha \in [k]} W_{d_{B_{\alpha,l_\alpha}},\hat{l}_\alpha},
\end{align*}
\item Repeat Step 1, until all users are satisfied.
\end{enumerate}

\begin{claim}
\label{claim:delivery_phase_correctness}
The delivery phase algorithm proposed above terminates and allows each user's demand to be satisfied. Furthermore the rate of transmission of the server is $R = q-1$.
\end{claim}
The proof of the above claim appears in the Appendix.

\section{Comparison with Existing Schemes}
\label{sec:comp}
In this section, we compare our proposed scheme with the scheme of \cite{Ma14} with $K=qk$ and $\frac{M}{N}=\frac{1}{q}$.
Let $R^*$ and $F^*$ denote the rate and the subpacketization level of our proposed scheme, $R^{MN}$ and $F^{MN}$ denote the rate and the subpacketization level of \cite{Ma14}, where
\begin{align*}
R^{MN}&=\frac{K(1-\frac{M}{N})}{1+\frac{KM}{N}}\\
&=\frac{qk-k}{1+k},\text{~and}\\
F^{MN}&=\binom{K}{\frac{KM}{N}}\\
&=\binom{qk}{k}.
\end{align*}

In our proposed scheme,
\begin{align*}
R^*&=q-1,\text{~and}\\
F^*&=q^{k-1}.
\end{align*}

Thus, the following conclusions can be drawn.
\begin{align*}
\frac{R^{MN}}{R^*}&=\frac{k}{1+k}.
\end{align*}
This implies that the rate of our proposed scheme and the scheme of \cite{Ma14} is almost the same for large $k$.

For large $k$, we show (see Appendix, Lemma \ref{lemma:subpack_comp_1}) that
\begin{align*}
	\frac{F^{MN}}{F^*}\approx (\frac{q}{q-1})^{qk-k}.
\end{align*}
This implies that our subpacketization is {\it exponentially} smaller compared to the scheme of \cite{Ma14}.
Table \ref{Table:Compare} shows a precise numerical comparison when $q=2$, i.e., $M/N = 1/2$.
\begin{table}[t]
	\centering
	\label{Table:Compare}
	
	\begin{tabular}{cccccccc}
		\hline
		\hline
		$K$ & 2     & 4     & 8    & 10     & 12      & 14     &16\\
		\hline
		\hline
		$R^{MN}$  & 0.67  & 0.75  &0.8   & 0.83   & 0.86    & 0.875  &0.89\\
		\hline
		$R^*$      &1      &1      &1     &1       &1        &1       &1\\
		\hline
		$F^{MN}$  &6      &20     &70    &252     &924      &3432    &12870\\
		\hline
		$F^*$      &2      &4      &8     &16      &32       &64      &128\\
		\hline
		\hline
	\end{tabular}
	
	\caption{Numerical comparison for $\frac{M}{N}=\frac{1}{2}$ and different values of $K$}.
\end{table}

An alternate technique for achieving $M/N = 1/q$ with a lower subpacketization level is to perform memory-sharing between
appropriate points using the scheme of \cite{Ma14}. Next, we compare our proposed scheme with memory-sharing for the case of $q=2$, i.e., the number of users $K = 2k$.


Towards this end, we divide each file into two smaller files $W_{n}^1, W_{n}^2$ with equal size, and further split $W_{n}^1, W_{n}^2$ into $\binom{2k}{t}$ and  $\binom{2k}{2k-t}$ subfiles, respectively, where $t< k$. In the placement phase, $\frac{t}{2k}$ fraction of the subfiles of $W_{n}^1$ and $\frac{2k-t}{2k}$ fraction of the subfiles of $W_{n}^2$ are placed in each user using the placement scheme of \cite{Ma14}. Thus, the overall cache at each user is $M = N \cdot \frac{1}{2} (\frac{t}{2k} + \frac{2k-t}{2k})$ so that $M/N = 1/2$. In the delivery phase, the transmission rate is given by
$$
R^{MN,MS}=\frac{1}{2}\bigg{(}\frac{2k-t}{1+t}+\frac{t}{1+2k-t}\bigg{)}.
$$
The subpacketization level of this scheme $F_s^{M-D,MS}=2\binom{2k}{t}$.
We compare our proposed scheme with the memory sharing scheme considered above by choosing a value of $t$ so that the rates of the schemes are approximately the same. Since $\frac{2k-t}{1+t}>1$ and $\frac{t}{1+2k-t}<1$, we approximate $R^{MN,MS}\approx\frac{2k-t}{2(1+t)}$. Then $R^{MN,MS}=R^*$ if $t=\frac{2k-2}{3}$. For this setting we show (see Appendix, Lemma \ref{lemma:subpack_comp_2}) that
\begin{align*}
\frac{F^{MN,MS}}{F^*}\approx 2^{2.8k}.
\end{align*}
Thus, our scheme has a significantly lower subpacketization level compared to a memory-sharing solution.
\section{Conclusions and Future Work}
\label{sec:conclusion}
In this work, we proposed a novel scheme for coded caching whose subpacketization level is exponentially smaller than the scheme in \cite{Ma14}. Moreover, for large number of users, the rate of our scheme is almost the same as \cite{Ma14}. Our schemes are derived from resolvable block designs generated by a single parity-check code over $\mathbb Z_q$.

There are several opportunities for future work. Our proposed scheme currently only works when $M/N$ is the reciprocal of a positive integer. Furthermore, even though our subpacketization level is significantly lower than \cite{Ma14}, it still scales exponentially with the number of users, albeit much slowly. Investigating schemes with subpacketization levels that grow sub-exponentially with $K$ and schemes that work with general values of $M/N$ are interesting directions for future work.

\bibliographystyle{IEEETran}
\bibliography{coded_caching,caching_refs,refs}
\appendix

\subsection*{Proof of Claim \ref{claim:intersect}}
Following the construction in Section \ref{sec:resolv_design_constr}, we note that a block $B_{i,l} \in \calP_i$ is specified by
$$
B_{i,l} = \{j : \bfT_{i,j} = l\}
$$
Now, consider $B_{i_1, l_1}, \dots, B_{i_{k-1}, l_{k-1}}$ (where $i_j \in [k], l_j \in \{0, \dots, q-1\}$) that are picked from $k-1$ distinct parallel classes $\calP_{i_1}, \dots, \calP_{i_{k-1}}$. W.l.o.g. we assume that $i_1 < i_2 < \dots < i_{k-1}$. Let $\calI =  \{i_1, \dots, i_{k-1}\}$ and $\bfT_{\calI}$ denote the submatrix of $\bfT$ obtained by retaining the rows in $\calI$. We will show that the vector $[l_1~l_2~\dots~l_{k-1}]^T$ is a column in $\bfT_{\calI}$.

To see this first we consider the case that $\calI = \{1, \dots, k-1\}$. In this case, the message vector $\bfu = [l_1~l_2~\dots~l_{k-1}]$ is such that $[\bfu \bfG_{SPC}]_{\calI} = [l_1~l_2~\dots~l_{k-1}]^T$ so that $[l_1~l_2~\dots~l_{k-1}]^T$ is a column in $\bfT_{\calI}$. On the other hand if $k \in \calI$, then we have $i_{k-1} = k$. Now, consider the system of equations in variables $u_1, \dots, u_{k-1}$.
\begin{align*}
u_{i_1} &= l_1,\\
u_{i_2} &= l_2,\\
\mathrel{\makebox[\widthof{=}]{\vdots}}\\
u_{i_{k-2}} &= l_{k-2},\\
u_1 + u_2 + \dots + u_{k-1} &= l_{k-1}.
\end{align*}
It is evident that this system of $k-1$ equations in $k-1$ variables has a unique solution over $\mathbb Z_q$. The result follows. \endproof

\subsection*{Proof of Claim \ref{claim:delivery_phase_correctness}}

In the arguments below, for the sake of convenience we argue that user $U_{B_{1,0}}$ can recover all its missing subfiles. As the delivery phase algorithm is symmetric with respect to the users, this equivalently shows that all users can recover their missing subfiles.

Note that $|B_{1,0}| = q^{k-2}$. Thus, user $U_{B_{1,0}}$ needs to obtain $q^{k-1} - q^{k-2}$ missing subfiles. The delivery phase scheme repeatedly picks $k$ users from different parallel classes $U_{B_{1, 0}}, U_{B_{2, l_2}}\dots, U_{B_{k, l_k}}$  such that $B_{1,0} \bigcap \cap_{i=2}^k B_{i,l_i} = \phi$. According to the equation transmitted in Step 3 of the algorithm, this allows $U_{B_{1, 0}}$ to recover subfile $W_{d_{B_{1,0}}, \hat{l}_1}$ where $\hat{l}_1 = \cap_{i=2}^{k} B_{i,l_i}$. Note that $U_{B_{1,0}}$ does not have $W_{d_{B_{1,0}}, \hat{l}_1}$ in its caches since $B_{1,0} \bigcap \cap_{i=2}^k B_{i,l_i} = \phi$.

Next, we count the number of equations that $U_{B_{1,0}}$ participates in. We can pick $k-2$ users from parallel classes $\calP_2, \dots, \calP_{k-1}$. Claim \ref{claim:intersect} ensures that blocks corresponding to these users intersect in a single point. Next we pick a block from the remaining parallel class $\calP_k$ such that the intersection of all the blocks is empty; this can be done in $q-1$ ways. Thus, there are a total of $q^{k-2}(q-1) = q^{k-1} - q^{k-2}$ equations in which user $U_{B_{1,0}}$ participates.

We have previously argued that each such equation allows $U_{B_{1,0}}$ to decode a subfile that it does not have in its cache. If we can argue that each equation provides a distinct file part then our argument is complete. Towards this end suppose that there exist sets of blocks $\{B_{2, l_2}, \dots, B_{k, l_k}\}$ and $\{B_{2, l_2'}, \dots, B_{k, l_k'}\}$ such that  $\{B_{2, l_2}, \dots, B_{k, l_k}\} \neq \{B_{2, l_2'}, \dots, B_{k, l_k'}\}$, but $\cap_{i=2}^k B_{i,l_i} = \cap_{i=2}^k B_{i,l_i'} = \{\beta\}$ for some $\beta \in [q^{k-1}]$. This is a contradiction since this in turn implies that $\cap_{i=2}^k B_{i,l_i} \bigcap \cap_{i=2}^k B_{i,l_i'} = \{\beta\}$, which is impossible since two blocks from the same parallel class have an empty intersection.

Finally, we calculate the rate of the delivery phase algorithm. We transmit a total of $q^{k-1}(q-1)$ equations, where each symbol is of size $F/q^{k-1}$. Thus, the rate is given by,
\begin{align*}
R &= q^{k-1}(q-1)\frac{F}{q^{k-1}}\\
&= (q-1)F.
\end{align*}
\endproof

\begin{lemma}
	\label{lemma:subpack_comp_1}
    Suppose $K=qk$ and let $q$ be fixed. Then,
    \begin{align*}
    \lim_{k\to \infty} \frac{1}{kq}\log_2 {\frac{F_s^{MN}}{F_s^*}}&=\bigg{(}1-\frac{1}{q}\bigg{)}\log_2 \bigg{(}\frac{q}{q-1}\bigg{)}.
    \end{align*}
\end{lemma}
	\begin{proof}
		It is well known \cite{graham1994concrete} that for $0\le p\le 1$ 
		\begin{equation}
		\label{equation:binom_approx}
				\lim_{n\to \infty} \frac{1}{n}\log_2 \binom{n}{pn}=H(p),
		\end{equation}
		where $H(\cdot)$ represents the binary entropy function. Using this
		\begin{align*}
		\lim_{k\to \infty}\frac{1}{qk}\log_2\frac{F_s^{MN}}{F_s^*}&=\lim_{k\to \infty} \frac{1}{qk}\log_2 \frac{\binom{qk}{k}}{q^{k-1}}\\
        &=H(\frac{1}{q})-\frac{\log_2 q}{q}\\
        &=\bigg{(}1-\frac{1}{q}\bigg{)}\log_2 \bigg{(}\frac{q}{q-1}\bigg{)}.
		\end{align*}
		
	\end{proof}
	
	\begin{lemma}
		\label{lemma:subpack_comp_2}
		Suppose $K=2k$, $\frac{M}{N}=\frac{1}{2}$ and $t=\frac{2k-2}{3}$. Then,
		\begin{align*}
			\lim_{k\to \infty}\frac{1}{k}\log_2\frac{F_s^{MN,MS}}{F_s^*}=&2.8.
		\end{align*}
	\end{lemma}
	\begin{proof}
This follows from eq. (\ref{equation:binom_approx}) in Lemma \ref{lemma:subpack_comp_1} and basic algebraic manipulations.
	\end{proof}
\end{document}